\documentclass{article}
\usepackage{cite}
\usepackage{amsmath,amssymb,amsfonts}
\usepackage{algorithmic}
\usepackage{graphicx}
\usepackage{textcomp}
\usepackage[usenames,dvipsnames,svgnames]{xcolor}
\usepackage{mdframed}

\usepackage{algorithm}
\usepackage{color,soul}
\usepackage{amsthm}
\usepackage{subcaption}
\usepackage{standalone}
\usepackage{tikz}
\usepackage{bm}
\usetikzlibrary{positioning}
\newtheorem{theorem}{Theorem}[section]
\newtheorem{definition}{Definition}[section]
\newtheorem{corollary}{Corollary}[section]

\newcommand{\eat}[1]{}
\newcommand{\ignore}[1]{}

\def\CU{\mathcal{U}}
\def\CM{\mathcal{M}}
\def\CT{\mathcal{T}}
\def\CB{\mathcal{B}}
\def\FB{\mathfrak{B}}
\def\pubfl{\mbox{puBFMap}}

\DeclareRobustCommand{\edit}[1]{{\sethlcolor{Yellow}\hl{#1}}}

\def\BibTeX{{\rm B\kern-.05em{\sc i\kern-.025em b}\kern-.08em
    T\kern-.1667em\lower.7ex\hbox{E}\kern-.125emX}}
\begin{document}
\title{Dynamic Partition Bloom Filters: A Bounded False Positive Solution For Dynamic Set Membership\\
(Extended Abstract)
}

\author{Sidharth Negi\footnote{Department of Computer Science and Engineering, IIT Delhi.} \and Ameya Dubey\footnotemark[1] \and Amitabha Bagchi\footnotemark[1] \and  Manish Yadav\footnotemark[1] \and Nishant Yadav\footnote{College of Information and Computer Sciences, Univ. of Massachusetts, Amherst.} \and Jeetu Raj\footnote{University of Illinois, Urbana-Champaign.}}

\maketitle
\thispagestyle{plain}
\pagestyle{plain}

\begin{abstract}
	Dynamic Bloom filters (DBF) were proposed by Guo et. al. in 2010 to tackle the situation where the size of the set to be stored compactly is not known in advance or can change during the course of the application. We propose a novel competitor to DBF with the following important property that DBF is not able to achieve: our structure is able to maintain a bound on the false positive rate for the set membership query across all possible sizes of sets that are stored in it. The new data structure we propose is a dynamic structure that we call Dynamic Partition Bloom filter (DPBF). DPBF is based on our novel concept of a Bloom partition tree which is a tree structure with standard Bloom filters at the leaves. DPBF is superior to standard Bloom filters because it can efficiently handle a large number of unions and intersections of sets of different sizes while controlling the false positive rate. This makes DPBF the first structure to do so to the best of our knowledge. We provide theoretical bounds comparing the false positive probability of DPBF to DBF. Our extensive experimental analysis demonstrates that our proposed structure takes up to three orders of magnitude lower time than DBF to process queries while keeping the false positive probability bounded unlike DBF.
\end{abstract}

%Bloom filters, Dynamic sets, Compression, False positives, Union, Intersection

\section{Introduction}
%http://cs.stanford.edu/people/chrismre/cs345/rl/writing-hints.pdf

% Context of the paper.
Bloom filter (BF)~\cite{bloom1970space} and its several static variants are widely used for compact set representation and efficient membership queries. They achieve compact representation at the cost of false positives in membership queries. The application designer has to decide an acceptable threshold for this false positive rate and provision the BFs in advance to ensure that the threshold is not violated. These structures allow efficient membership query, set intersection and union operations. However, once the set is stored in the BF, it is no longer possible to make any changes which leads to problems in highly dynamic scenarios.

% Problem definition.
The sets often need to be stored in highly dynamic scenarios, and applications such as Bloom joins on distributed databases~\cite{mackert1986bloomjoin, guo2006theory}, and informed routing and global collaboration in unstructured P2P networks~\cite{kubia2000oceanstore, ledlie2002organize, acuna2003planetp} may require storing sets that differ greatly in size. Also, it may not always be possible to have the knowledge of the set size in advance. In such scenarios, choosing the right size for the Bloom filter poses a significant challenge. A large Bloom filter size causes unnecessary space overhead, while small Bloom filters lead to undesirably high false positive rates. Moreover, many applications \cite{broder2004network, guo2006theory} may also require intersection and union operation on sets stored in the Bloom filters. \\

% Existing solution 1.
With these issues in mind Guo et. al. proposed the Dynamic Bloom filter (DBF) in \cite{guo2010dynamic}. Rather than a single bit-array, the DBF is a list of Standard Bloom Filters. The false positive rate of each SBF used as a unit in the DBF is maintained below a pre-defined threshold. The DBF starts as a list containing a single unit Bloom Filter and as elements of the set are inserted into the DBF, they get populated into the last Bloom filter in the list. When the estimated false positive probability of the last Bloom filter reaches the pre-defined threshold, another empty unit Bloom filter is appended to the end of the list. This method works well to handle varying set sizes but there is a critical flaw: the false positive rate increases {\em linearly} with the set size. This fact was not highlighted in Guo et. al.'s paper, but we provide a mathematical proof for this fact (see Appendix~\ref{sec:theo_ana}).

% Existing solution 2 and its limitation.
%\ABcomment{Let us completely omit inverted indexes. Our thrust is: We beat DBF. It is not, we beat inverted indexes. Mainly because we cannot.}

%Inverted index is another solution which has long been used in search engines for set intersections and has been optimized for quick intersection operations. It works well in cases where latency is acceptable in set update realization, e.g. in web page indexing done by search engines, it is acceptable for changes in a web page to show up within a span of some hours. However for applications in social network mining such as finding the set of users who have posted a particular hashtag, we can't afford to have latency in set updation but we still need efficient intersection operation on sets of varying sizes. Inverted index fails to be of use in such cases. \\

\begin{figure}[h!]
\centering
\begin{minipage}{0.48\textwidth}
		\includegraphics[scale=0.6]{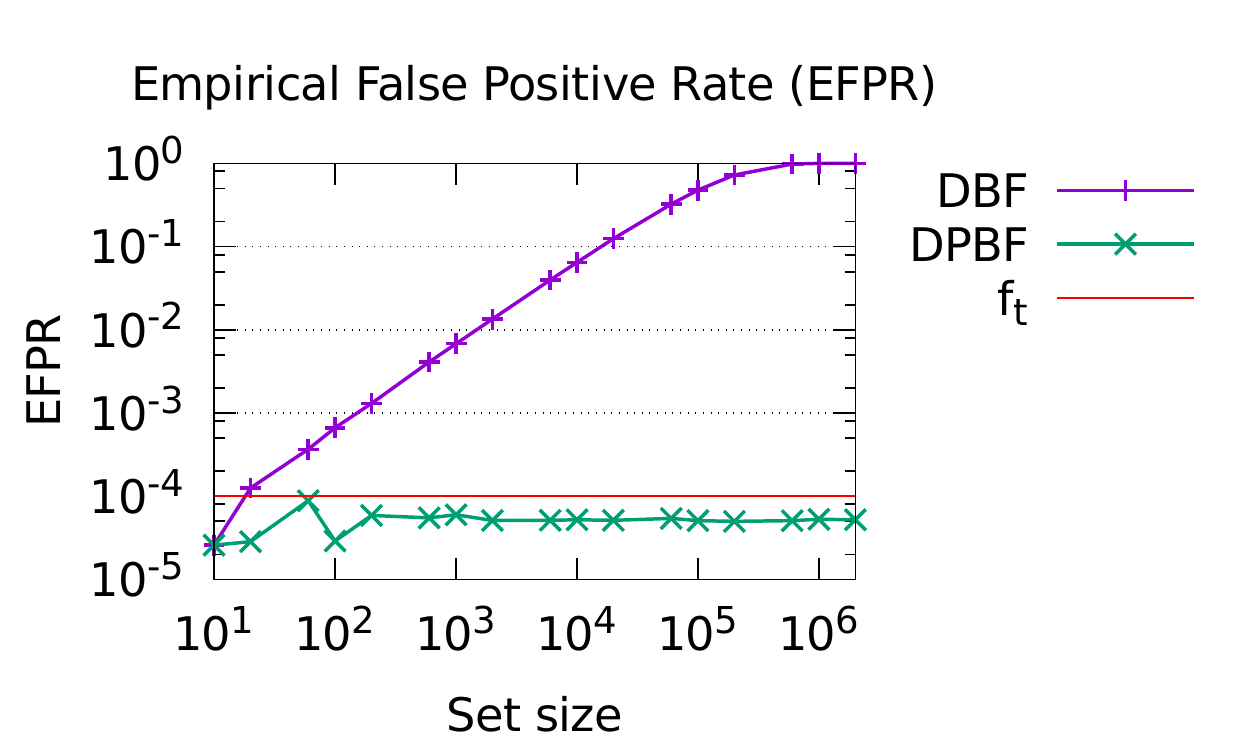}
		\caption{DBF exceeds the threshold ($p_t=10^{-4}$) even for smaller sets.}
		\label{fig:intro_dbf1}
\end{minipage}\hfill
\begin{minipage}{0.48\textwidth}
		\includegraphics[scale=0.6]{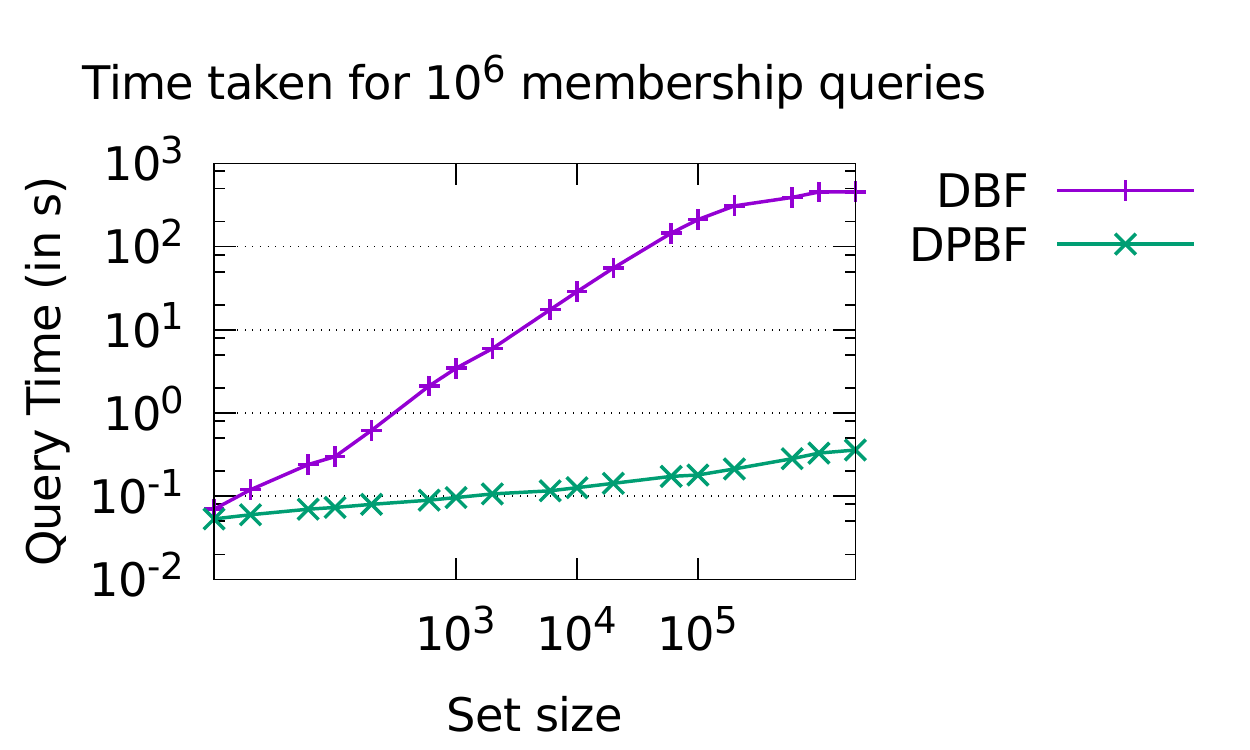}
		\caption{DBF takes over 2 orders of magnitude more time than DPBF.}
		\label{fig:intro_dbf2}
\end{minipage}
\label{fig:intro_dbf}
\end{figure}

Our proposal the {\em Dynamic Partition Bloom Filter} (DPBF) is designed to overcome this shortcoming: Given a desired bound on the false positive rate by the application designer, DPBF can maintain it no matter what the size of the set stored in it, and no matter how widely this size varies through the life of the application. In Figure~\ref{fig:intro_dbf1} we demonstrate this effect by comparing the false positive rate across DPBF and DBF. We set a false positive threshold of $10^{-4}$ and vary set sizes across six orders of magnitude and empirically measure the false positive achieved by the two structures. DPBF holds its line while DBF shows a linear increase, dominating DPBF even for small sizes. This experiment is presented in greater detail in the full version of this paper which is under submission.

Despite being able to maintain false positive rate, DPBF performs very well compared to DBF on query time as well. In Figure~\ref{fig:intro_dbf2} we see that the query time for DPBF is 2 orders of magnitude better than DBF. The reason for this is that the DPBF is a binary tree-like structure with unit Standard Bloom Filters at the leaves. As such it is natural that it has a great query advantage over the linear structured DBF. Creating and maintaining this structure requires more space than DBF takes but in the full version of this paper we show that the extra space committed is not significant.

This extended abstract is an early announcement of our work in which we describe our new data structure. A fuller version of this paper that contains implementation details and proofs is under submission.

\section{Dynamic Partition Bloom Filters}
\label{sec:dpbf}

We now define the Dynamic Partition Bloom filter in detail. But to do so we define a tree-based hierarchical partitioning scheme for the namespace that we call a partition tree and explain how to add Bloom filters to this scheme in preparation for our main definition.

\subsection{Preliminaries: Partition Trees}
\label{sec:dpbf:partition}

Assume for simplicity of exposition that our namespace $\mathcal{U}$ contains ids $\{0, 1, \ldots, |\mathcal{U}|-1\}$. In the case of a general namespace we can always map that namespace to contiguous integers beginning with 0.
\begin{definition}
\label{def:pt}
The {\em Partition Tree of depth $d$} associated with $\mathcal{U}$ is a complete binary tree $\CT_d(\mathcal{U})$ of depth $d$ and a mapping of nodes of this tree to subsets of $\mathcal{U}$. To describe the mapping, let us say that the $j$th node of level $i$ of $\CT_d(\CU)$ is $N_{i,j}$. Then the subset of $\CU$ associated with $N_{i,j}$ is 
\[\CM_{i,j} = \left\{\ell\ : j \cdot \frac{|\CU|}{2^i} \leq \ell < (j + 1) \cdot \frac{|\CU|}{2^i}\right\}. \]
\end{definition}
We note that all the subsets of $\CU$ associated with a given level of the Partition Tree are equal in size and form a partition of $\CU$. Further the subsets associated with the two children of each internal node of $\CT_d(\CU)$ form a partition of the subset associated with that node. In other words, the partition at level $i+1$ refines the partition at level $i$ all the way down to the leaf level. The partition at the leaf level is the most fine grained while the partition at the root contains only one set, the entire namespace $\CU$.

Now we associate SBFs with the Partition Tree to give us what we call a Bloom Partition Tree. We need an additional parameter here: the FPR $f$ that we want to maintain. We also note that given a target FPR $f$, if we want to insert a set of size $n$ into an SBF with $k$ hash functions then from (\ref{eq:SBF-FPR}) we can back calculate the number of bits we need to allocate as
\begin{equation}
    m(n,k,f) = \left\lceil -\frac{nk}{\ln{(1 - f^{\frac{1}{k}})}} \right\rceil
    \label{eq:2}
\end{equation}
With this in hand we are ready to define the Bloom Partition Tree.
\begin{definition}
Given a namespace $\CU$, and a target FPR of $f$, the {\em Bloom Partition Tree of depth $d$} associated with $\CU$ that maintains target FPR $f$ is a Partition Tree of depth $d$, $\CT_d(\CU)$ along with a homogenous set of $2^d$ SBFs, $\FB= \{\CB_0, \ldots,\CB_{2^d -1}\}$, with $\CB_i$ associated with leaf $N_{d,i}$ of $\CT_d(\CU)$. If these SBFs have $k$ hash functions associated with them then the number of bits allocated to each SBF is $m(|\CU|/2^d, k, f)$ where the function $m(\cdot,\cdot,\cdot)$ is as defined in (\ref{eq:2}). 
\end{definition}
We will refer to the $2^d$ SBFs of the Bloom Partition Tree of depth $d$ associated with $\CU$ as the {\em unit Bloom Filters} ({\bf uBF}) of the Bloom Partition Tree. 

We are now ready to present the Dynamic Partition Bloom Filter structure, but before we do so we note that the Partition Tree and Bloom Partition Tree defined above are {\em not} to be stored in memory. These are defined here to help understand the DPBF and what it actually stores in memory.

\subsection{Definition: Dynamic Partition Bloom Filters}
\label{sec:dpbf:dpbf}

The DPBF comprises two substructures, a hash map of populated SBFs, that we call the {\em Populated Unit Bloom Filter Map} ({\bf puBFMap}) and a compressed version of the Bloom Partition Tree with populated SBFs at its leaves that we call the {\em Compressed Populated Bloom Partition Tree} ({\bf CPBPT}). Membership queries are answered by the CPBPT while the puBFMap is used to realise the insertion, union and intersection operation. 

\paragraph{Populated Unit Bloom Filter Map}
An SBF $\CB$ is said to be {\em populated} if some set $S$ has been stored in it. We denote an SBF $\CB$ populated with set $S$ by $\CB(S)$.
\begin{definition}
Given an $A \subseteq \CU$ and a Bloom Partition Tree $\CT_d(\CU) \cup \FB$, where $\FB = \{ \CB_0, \ldots, \CB_{2^d-1}\}$, the {\em Populated Unit Bloom Filter Map} of $A$ is a set of populated SBFs
\[ \pubfl(A) = \{ \CB_i(A \cap N_{d,i}) : 0 \leq i \leq 2^d -1 , A \cap N_{d,i} \ne \emptyset \}.\]
\end{definition}
To restate the definition in plain language: We intersect the set $A$ with all the subsets of $\CU$ defined by the most fine grained partition of the BPT, i.e. the partition at the leaves. All the non-empty intersections are stored in unit Bloom Filters and this makes up the puBFMap of $A$.

\begin{figure}[htbp]
\begin{center}
  \begin{tikzpicture}[
    every label/.append style={font=\Large},
    edge from parent/.style = {draw, -latex, line width=1mm},
    level/.style = {level distance = 1cm},
    level 1/.style = {sibling distance = 7cm},
    level 2/.style = {sibling distance = 3.5cm},
    level 3/.style = {sibling distance = 1.75cm, child anchor = north},
]
\tikzstyle{internal} = [circle, draw, thin, color=black, inner sep=7pt]
\tikzstyle{leaf} = [rectangle, align=center, draw=black, thick, font=\large]
\tikzstyle{edgel} = [font=\large]
\node(params)[font=\Large]{$|\mathcal{U}|=32, n_t=4, k=2, m=5$};
\node(set)[font=\Large, below=0.5pt of params]{$A = \{4,5, 8,\textcolor{red}{10}, 17, 19, 22, \textcolor{SeaGreen}{25}, 31 \}$};
\node[internal, below=1pt of set]{}
    child {node[internal]{} 
        child {node[internal]{} 
            child{node(0)[leaf, label={[label distance=20pt]below:$N_{3,0}$}, label=below:$\{ \}$]{\texttt{0\ldots 0}}
            }
            child{node[leaf, label={[label distance=20pt]below:$N_{3,1}$}, label=below:$\{ \}$]{\texttt{0\ldots 0}} 
            }
            edge from parent[color=black]
        }
        child {node[internal]{}
            child{node(1)[leaf,label={[label distance=20pt]below:$N_{3,2}$}, color=black, label=below:$\{ \bm{10}\}$]{\texttt{00110}}
            edge from parent[color=red]
            }
            child[color=black]{node[leaf, label={[label distance=20pt]below:$N_{3,3}$}, label=below:$\{ \}$]{\texttt{0\ldots 0}} edge from parent[color=black]
            }
            edge from parent[color=red]
        }
        edge from parent[color=red]
    }
    child {node[internal]{} 
        child {node[internal]{} edge from parent[color=black]
            child{node[leaf, label={[label distance=20pt]below:$N_{3,4}$}, , label=below:$\{ \}$]{\texttt{0\ldots 0}}
            }
            child{node[leaf,label={[label distance=20pt]below:$N_{3,5}$}, , label=below:$\{ \}$]{\texttt{0\ldots 0}}
            }
        }
        child {node[internal]{}
            child{node(2)[leaf, label={[label distance=20pt]below:$N_{3,6}$}, color=black, label=below:$\{ \bm{25} \}$]{\texttt{01010}}
            }
            child[color=black]{node(last)[leaf, label={[label distance=20pt]below:$N_{3,7}$}, , label=below:$\{ \}$]{\texttt{0\ldots 0}} edge from parent[color=black]
            }
        }
        edge from parent[color=SeaGreen]
    };

%\node(head)[draw, inner sep=6pt, below left=of 0]{};
\node(1shadow)[leaf, below=45pt of 1, label=below:2]{\texttt{00110}};
\node(2shadow)[leaf, below=45pt of 2, label=below:6]{\texttt{01010}};
%\node(tail)[draw, inner sep=6pt, below right=of last]{};
%\draw (tail.north east) -- (tail.south west);
%\draw (tail.north west) -- (tail.south east);
\draw [-latex, very thick] (1shadow) edge (2shadow) ;
%(2shadow) edge (tail);
\end{tikzpicture}%     without .tex extension
  % or use \input{mytikz}
\end{center}
 \caption{Bloom Partition Tree being populated by elements of $A$, and the corresponding respresentation, puBFMap, which is actually stored in memory.}
  \label{fig:pBF-tree-list}
\end{figure}
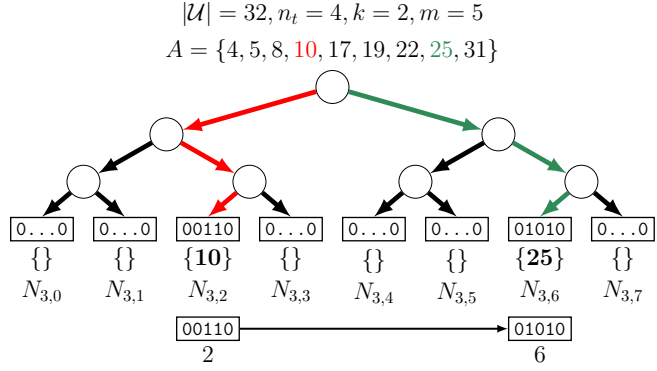

We illustrate the concept of the puBFMap with an example. In Figure~\ref{fig:pBF-tree-list} we consider a case where $|\CU| = 32$. A Bloom Partition Tree of depth 3 is being used, so $n_t = 4$ and there are 8 uBFs at the leaves. Each uBF has 5 bits in it and has 2 hash functions associated with it. We populate this BPT with the set  $A=\{4,5,8, 10, 17, 19, 22, 25, 31\}$.  We see that each element $x \in A$ is inserted into the unit Bloom filter of the leaf node corresponding to its partition. For example, consider the element 10. We find the partition to which the element 10 belongs, as represented by the red path. Using Definition \ref{def:pt}, \[2 \cdot \frac{32}{2^3} \leq 10 < (2 + 1) \cdot \frac{32}{2^3}\] so the element $10$ is inserted into the leaf with $j = 2$ of the last level, i.e. Bloom filter at $N_{3,2}$.

\paragraph{Compressed Populated Bloom Partition Tree}
We now turn to the CPBPT. Given a set $A$ and a membership query for some $x \in \CU$, we can see that the Bloom Partition Tree can act as a Binary Search Tree and guide us to the leaf $N_{d,i}$ such that $x \in \CM_{d,i}$. Now if $\pubfl(A)$ contains a populated version of $\CB_i$ then an SBF query can reveal whether $x$ is in $A$ or not. If $\pubfl(A)$ does not contain $\CB_i$ then a negative answer can be given directly. However, the input elements may be distributed over the namespace in such a way that each uBF stores only a few elements, thus wasting a lot of memory per uBF. So, we store a compressed version of this structure which we call CPBPT. 

We now define this structure. But first we introduce some notation: Given an SBF $\CB$ with $m$ bits and an FPR $f$, the {\em target population} $n_t$ is the maximum number of elements that can be stored in $\CB$ while maintaining an FPR of at most $f$. Note that $n_t$ can be calculated from (\ref{eq:2}) by placing the given value of $m$ on the LHS and solving for $n$. Also observe that if we are working with a BPT of depth $d$ for a given FPR $f$, we have chosen the size of the uBFs such that $n_t(f) = |\CU|/2^d$, i.e., $n_t(f)$ is the size of the maximum possible set that can be stored in any uBF.

\begin{definition}
Given a set $A$, a target FPR $f$ and a BPT $\CT_d(U)$, the {\em Compressed Bloom Partition Tree} is a obtained by associating uBFs with the leaves of a subtree $\CT_d(U,A,f)$ of $\CT_d(U)$ defined as follows:
\begin{itemize}
\item $N_{0,0}$ is the root of $\CT_d(U,A,f)$.
\item For all $i > 0$ and $0 \leq j < 2^i$, $N_{i,j}$ is a leaf of $\CT_d(U,A,f)$ if $|\CM_{i,j} \cap A| \leq n_t(f)$ but $|\CM_{i-1,\lfloor j/2\rfloor} \cap A| > n_t(f)$
\end{itemize}
If $N_{i,j}$ is a leaf of $\CT_d(U,A,f)$ we associate uBF $\CB_{i,j}$ of size $m(n_t(f),k,f)$ with $N_{i,j}$ and populate it with the set $\CM_{i,j} \cap A$.
\end{definition}
The easiest way to understand the CPBPT is algorithmically: Create the puBFMap of $A$ by populating the uBFs at the leaf level of the BPT. We are guaranteed that each leaf node has at most $n_t(f) = |\CU|/2^d$ elements associated with it. If a leaf and its sibling together still have at most $n_t(f)$ elements we can merge them into their and maintain a single uBF that stores the elements associated with the union. This process can continue till we reach a compressed version of the BPT with the property that every internal node has the property that the number of elements of $A$ associated with its subset of $\CU$ exceeds $n_t(f) = |\CU|/2^d$ and every leaf has the property that the number of elements of $A$ associated with its subset of $\CU$ are at most $n_t(f)$.

\begin{figure}[htbp]
\begin{center}
  \def\dis{20pt}
\begin{tikzpicture}[
    every label/.append style={font=\Large},
    edge from parent/.style = {draw, -latex, line width=1mm},
    rev/.style= {edge from parent/.style = {draw, latex-, line width=1.5mm, orange}},
    level/.style = {level distance = 1cm},
    level 1/.style = {sibling distance = 7cm},
    level 2/.style = {sibling distance = 3.5cm},
    level 3/.style = {sibling distance = 1.75cm, child anchor = north},
]
\tikzstyle{internal} = [circle, draw, thin, color=black, inner sep=7pt]
\tikzstyle{leaf} = [rectangle, align=center, draw=black, thick]
\node[internal, label=45:$N_{0,0}$]{}
    child {node[internal, label=135:$N_{1,0}$]{}
        child[rev] {node[internal, label=135:$N_{2,0}$]{}
            child{node[leaf, label=below:$\bm{\{ \}}$, label={[label distance = \dis]below:$N_{3,0}$}]{\texttt{0\ldots 0}}
            }
            child{node(1)[leaf, label=below:$\bm{\{ 4,5\}}$, label={[label distance = \dis]below:$N_{3,1}$}]{\texttt{11010}}
            }
        }
        child[rev] {node[internal, label=45:$N_{2,1}$]{}
            child{node(2)[leaf, label=below:$\bm{\{8, 10\}}$, label={[label distance = \dis]below:$N_{3,2}$}]{\texttt{01110}}
            }
            child{node[leaf, label=below:$\bm{\{ \}}$, label={[label distance = \dis]below:$N_{3,3}$}]{\texttt{0\ldots 0}}
            }
        }
    }
    child {node[internal, label=45:$N_{1,1}$]{}
        child {node[internal, label=135:$N_{2,2}$]{}
            child[rev]{node(3)[leaf, label=below:$\bm{\{ 17,19\}}$, label={[label distance = \dis]below:$N_{3,4}$}]{\texttt{10101}}
            }
            child[rev]{node(4)[leaf, label=below:$\bm{\{ 22\}}$, label={[label distance = \dis]below:$N_{3,5}$}]{\texttt{01100}}
            }
        }
        child {node[internal, label=45:$N_{2,3}$]{}
            child[rev]{node(5)[leaf, label=below:$\bm{\{ 25\}}$, label={[label distance = \dis]below:$N_{3,6}$}]{\texttt{10010}}
            }
            child[rev]{node(6)[leaf, label=below:$\bm{\{ 31\}}$, label={[label distance = \dis]below:$N_{3,7}$}]{\texttt{10100}}
            }
        }
    };
\node(1shadow)[leaf, below=45pt of 1, label=below:1]{\texttt{11010}};
\node(2shadow)[leaf, below=45pt of 2, label=below:2]{\texttt{01110}};
\node(3shadow)[leaf, below=45pt of 3, label=below:4]{\texttt{10101}};
\node(4shadow)[leaf, below=45pt of 4, label=below:5]{\texttt{01100}};
\node(5shadow)[leaf, below=45pt of 5, label=below:6]{\texttt{10010}};
\node(6shadow)[leaf, below=45pt of 6, label=below:7]{\texttt{10100}};
\draw [-latex, very thick] (1shadow) edge (2shadow) (2shadow) edge (3shadow) (3shadow) edge (4shadow) (4shadow) edge (5shadow) (5shadow) edge (6shadow);
\end{tikzpicture}%     without .tex extension
  % or use \input{mytikz}
\end{center}
  \caption{Compressing the BPT after population.}
  \label{fig:compression}
\end{figure}
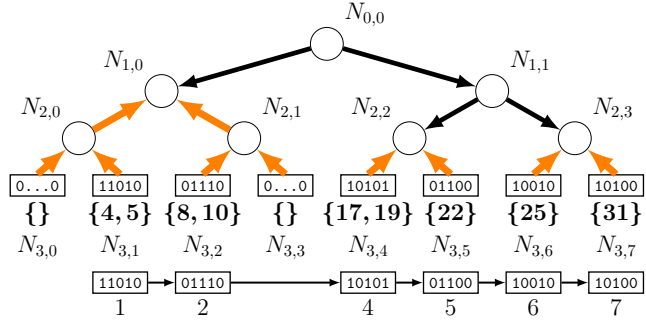
In Figure~\ref{fig:compression} we return to the example introduced in Figure~\ref{fig:pBF-tree-list} to illustrate the compression process. Since $|A \cap \CM_{1,0}| = 4 \leq 4$, the entire subtree rooted at $N_{1,0}$ is compressed into its root. Similarly the two children of $N_{2,2}$ can be compressed into $N_{2,2}$ since $|A \cap \CM_{2,2}| = 3 \leq 4$ and the two children of $N_{2,3}$ can be compressed into $N_{2,3}$ since $|A \cap \CM_{2,3}| = 2 \leq 4$ but the nodes $N_{2,2}$ and $N_{2,3}$ can't be merged since their parent's subset of the namespace, $\CM_{1,1}$ has an intersection of size $5 > 4$ with $A$.

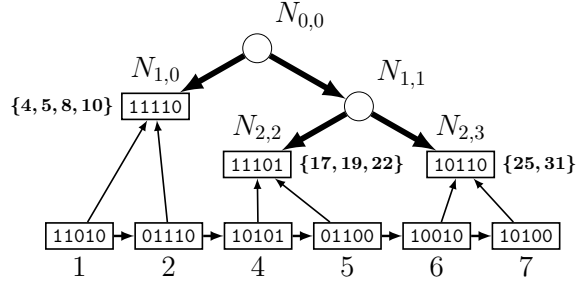
\begin{figure}[htbp]
  \centering
  \begin{tikzpicture}[
    every label/.append style={font = \Large},
    edge from parent/.style = {draw, -latex, line width=1mm},
    level/.style = {level distance = 1cm},
    level 1/.style = {sibling distance = 3.5cm},
    level 2/.style = {sibling distance = 3.5cm},
    %level 3/.style = {sibling distance = 1.75cm, child anchor = north}
]
\tikzstyle{internal} = [circle, draw, thin, black, inner sep=5pt]
\tikzstyle{leaf} = [rectangle, align=center, draw=black, thick]
\node(1)[internal, label=above right:$N_{0,0}$]{}
    child{node(a)[leaf, label=above:$N_{1,0}$, label=left:\small $\bm{\{4,5,8,10\}}$]{\texttt{11110}}
    }
    child {node[internal, label=above right:$N_{1,1}$]{}
        child{node(b)[leaf, label=above:$N_{2,2}$, label=right:\small $\bm{\{17,19,22\}}$]{\texttt{11101}}
        }
        child{node(c)[leaf, label=above:$N_{2,3}$, label=right:\small $\bm{\{25,31\}}$]{\texttt{10110}}
        }
    };
\node(1shadow)[leaf, below left=80pt and 65pt of 1, label=below:1]{\texttt{11010}};
\node(2shadow)[leaf, right=10pt of 1shadow, label=below:2]{\texttt{01110}};
\node(3shadow)[leaf, right=10pt of 2shadow, label=below:4]{\texttt{10101}};
\node(4shadow)[leaf, right=10pt of 3shadow, label=below:5]{\texttt{01100}};
\node(5shadow)[leaf, right=10pt of 4shadow, label=below:6]{\texttt{10010}};
\node(6shadow)[leaf, right=10pt of 5shadow, label=below:7]{\texttt{10100}};
\draw [-latex, very thick] (1shadow) edge (2shadow) (2shadow) edge (3shadow) (3shadow) edge (4shadow) (4shadow) edge (5shadow) (5shadow) edge (6shadow);
\draw [-latex, thick] (1shadow) edge (a) (2shadow) edge (a) (3shadow) edge (b) (4shadow) edge (b) (5shadow) edge (c) (6shadow) edge (c);
\end{tikzpicture}%     without .tex extension
  % or use \input{mytikz}
  \caption{The final DPBF, which consists of the CPBPT and the puBFMap. Note the added pointers from the puBFMap nodes to the corresponding leaves in the CPBPT.}
  \label{fig:snapshot}
\end{figure}
Finally in Figure~\ref{fig:snapshot} we see the composite DPBF comprising the puBFMap and the CPBPT for this example.

\paragraph{Space Complexity of DPBF}
For convenience of notation, we define
\begin{quote}
$s$: the number of uBFs in the puBFMap of the DPBF. 
\end{quote}
If we store a set $S$ in the DPBF built on a BPT of depth $d$ then 
\begin{equation}
\label{eq:s-lb}
s \geq \left\lceil \frac{|S| 2^d}{|\CU|}\right\rceil,
\end{equation}
since each uBF is allowed to store at most $|\CU|/2^d$ element of $S$. If $s$ is equal to the lower bound then the puBFMap has exactly the same number of Bloom Filters as DBF would have if DBF used Bloom Filters of the same size. In the worst case, however, $s$ can be as large as $2^d$ but since in most applications we store sets that are typically orders of magnitude smaller than the name space, we expect $s$ to be quite small. 

Note that $s$ is a natural upper bound on the number of uBFs in the CPBPT. Although the number of uBFs in the CPBPT could be far fewer than $s$ it is also possible, in the worst case, that if the uBFs of the puBFMap are highly populated then the number of uBFs in the CPBPT is also exactly $s$. Hence the total number of uBFs allocated is at most $2s$, each of size $m(|\CU|/2^d,k,f)$.

Additionally the internal nodes of the CPBPT are at most the number of leaves of the CPBPT since each internal node of the CPBPT has exactly 2 children, i.e., these take $O(s)$ space since the number of leaves of CPBPT can be at most $s$.

Hence, using (\ref{eq:2}), in terms of its parameters, we have in asymptotic terms that the space used by DPBF to store a set $S$ using $s$ uBFs is
\[\boxed{\theta \left(s \cdot \left\lceil -\frac{|\CU|\cdot k}{2^d \cdot \ln{(1 - f^{\frac{1}{k}})}} \right\rceil\right)}\]
bits. The quantity $s$ is a property of the set being stored in the DPBF. In the worst case if the elements of the set are distributed uniformly across the namespace $s$ can be as large as $2^d$. In practice however this is not the case. In the full version of the paper we show through experiments that the space taken by the DPBF is not significantly larger than that taken by DBF on real data, which is what we would expect if $s$ is close to the lower bound given in (\ref{eq:s-lb}).

\appendix

\section{The false positive probability of DBF grows linearly}
\label{sec:theo_ana}
\begin{theorem}
Let $\mathcal{BF_{DP}}(A)$ be DPBF with FPR threshold $f_t$, and $\mathcal{BF_{D}}(A)$ be DBF storing set $A \subseteq \mathcal{U}$ where both have unit Bloom filter of size $m$, $k$ hash functions, and at most $n_t = -\frac{m}{k}\ln \Big( 1 - f_t^{\frac{1}{k}} \Big)$ elements in any unit Bloom filter. The effective FPR $f_D$ in $\mathcal{BF_{D}}(A)$ is such that,
\[ f_D \geq 1 - (1 - f_t)^{\lfloor \alpha \rfloor} \approx \lfloor \alpha \rfloor f_t\]
where, $\alpha = \frac{|A|}{n_t} $
\end{theorem} 
\begin{proof}
For $\alpha < 1$, the theorem trivially holds as $\lfloor \alpha \rfloor = 0$. Also, $\mathcal{BF_{DP}}(A)$ and $\mathcal{BF_{D}}(A)$ consists of just 1 unit Bloom filter containing all $n_A = |A| (\leq n_t)$ elements of set $A$, and thus both have equal FPR.\\
For $\alpha \geq 1$, $\mathcal{BF_{D}}(A)$ will contain $\lceil \alpha \rceil$ unit Bloom filters its list. All unit Bloom filters but the last contain $n_t$ elements in them. So, FPR of each unit bloom filter but the last is given by ,
\[ f_{unit} = \Big(1 - \exp\big(\frac{-n_t k}{m}\big)\Big)^k = f_t \] 
The last unit Bloom filter in the list contain only $n_A - \lfloor \alpha \rfloor n_t $, hence its FPR is given by,
\[f_{last} = \Big( 1 -  \exp\bigg(  \frac{-(n_A - \lfloor \alpha \rfloor n_t)k }{m}\bigg) \Big)^k  \]
Membership queries in $\mathcal{BF_{D}}(A)$ are answered by probing each unit Bloom filter present in the list. Thus, $x \notin A$ can be returned as a false positive if any of the unit Bloom filters returns a positive result for it; and only when none of the unit Bloom filters return true for a query, can the membership query in $\mathcal{BF_{D}}(A)$ be answered in the negative. Thus,

\[ P\big( x \notin \mathcal{BF_{D}}(A) |\ x \notin A \big)  =  (1 - f_{unit})^{\lfloor \alpha \rfloor}\times ( 1 - f_{last})\]
\[ P\big( x \notin \mathcal{BF_{D}}(A) |\ x \notin A \big)  \leq  (1 - f_{unit})^{\lfloor \alpha \rfloor}\]

\[ \implies [ P\big( x \in \mathcal{BF_{D}}(A) |\ x \notin A \big)  \geq  1 -  (1 - f_{unit})^{\lfloor \alpha \rfloor}\]
\[ f_D \geq \lfloor \alpha \rfloor f_{unit} \quad [\because\ ( 1 - x)^y \approx 1 - xy] \]
\[ f_D \geq \lfloor \alpha \rfloor f_{t}\]
\end{proof}

\end{document}